\documentclass{article}

\usepackage{graphicx} 
\usepackage[hidelinks]{hyperref} 
\usepackage{natbib} 
\usepackage{subcaption} 
\usepackage{enumitem} 
\usepackage{pifont} 
\usepackage{comment}
\usepackage[margin=1in]{geometry} 

\usepackage{amssymb,amsmath,amsthm,amsfonts}
\usepackage{mathtools}
\usepackage{bbm}
\usepackage{color,xcolor}

\newcommand{\one}{\mathbbm{1}}
\newcommand{\E}{\operatorname{\mathsf{E}}}
\newcommand{\Q}{\operatorname{\mathsf{P}}}
\newcommand{\R}{\mathbb{R}}

\newcommand{\Aa}{\mathcal{A}}

\newtheorem{propo}{Proposition}

\theoremstyle{definition}
\newtheorem{defin}[propo]{Definition}
\newtheorem{remark}[propo]{Remark}

\definecolor{darkteal}{rgb}{0,0.35,0.35}

\title{Assessing the conditional calibration of interval forecasts using decompositions of the interval score}
\date{}
\author{Sam Allen\thanks{Institute of Statistics, Karlsruhe Institute of Technology, Karlsruhe, Germany. \url{sam.allen@kit.edu}} \and Julia Burnello\thanks{Independent researcher, \url{julia.burnello@gmail.com}} \and Johanna Ziegel\thanks{Seminar for Statistics, ETH Zurich, Zurich, Switzerland. \url{ziegel@stat.math.ethz.ch}}}

\begin{document}

\maketitle

\begin{abstract}
    Forecasts for uncertain future events should be probabilistic. Probabilistic forecasts are commonly issued as prediction intervals, which provide a measure of uncertainty in the unknown outcome whilst being easier to understand and communicate than full predictive distributions. The calibration of a $(1-\alpha)$-level prediction interval can be assessed by checking whether the probability that the outcome falls within the interval is equal to $1-\alpha$. However, such coverage checks are typically unconditional and therefore relatively weak. Although this is well known, there is a lack of methods to assess the conditional calibration of interval forecasts. In this work, we demonstrate how this can be achieved via decompositions of the well-known interval (or Winkler) score. We study notions of calibration for interval forecasts and then introduce a decomposition of the interval score based on isotonic distributional regression. This decomposition exhibits many desirable properties, both in theory and in practice, which allows users to accurately assess the conditional calibration of interval forecasts. This is illustrated on simulated data and in three applications to benchmark regression datasets.
\end{abstract}

\section{Introduction}\label{sec:intro}

Reliable forecasts are essential for effective decision making in the face of uncertainty. Forecasts are commonly issued as prediction intervals, which specify a range within which the outcome is expected to lie with a certain probability. Prediction intervals account for the uncertainty in the unknown outcome, in contrast to traditional point forecasts, and additionally serve as a convenient means to communicate the output from full predictive distributions to forecast users. They are therefore widely reported in a variety of application domains: national weather centres commonly issue interval forecasts for variables such as temperature and precipitation; epidemiological forecasts often take the form of a collection of prediction intervals at several probability levels, such as those available in the \cite{Covid19Hub}; while prediction intervals have also been requested in popular forecasting competitions \citep{Makridakis2020,Makridakis2022}. In the following, the terms \emph{interval forecast} and \emph{prediction interval} are used interchangeably.

Interval forecasts have gained considerable attention recently due to the development of conformal prediction algorithms that can generate prediction intervals that are calibrated (or valid) by construction \citep[Section 2]{VovkEtAl2022}. Generally, a $(1-\alpha)$-level prediction interval, for $\alpha \in (0, 1)$, is said to be calibrated if it contains the outcome with probability $1-\alpha$. This can be assessed in practice by checking that the empirical coverage of the prediction intervals is close to $1-\alpha$. It has repeatedly been acknowledged that such unconditional checks for calibration are insufficient, since biases can cancel out in such a way that the forecasts have the correct marginal coverage despite being clearly miscalibrated in particular circumstances \citep{Christoffersen1998}. For example, many conformal predictive algorithms achieve unconditional calibration despite issuing prediction intervals of a fixed length, which do not adapt to covariates or the complexity of the forecasting task \citep[see e.g. the discussion in][]{RomanoEtAl2019}.

It is therefore desirable to additionally assess the conditional calibration of interval forecasts. In practice, conditional calibration is typically evaluated by grouping together forecasts according to some additional information, and then assessing calibration separately for the forecasts in each of these groups. For example, interval forecasts could be assessed separately in different periods of time, or regions of space, or conditionally on aspects of the interval itself; \cite{FeldmanEtAl2021} and \cite{SebastianEtAl2024}, for example, propose conditioning on the length of the prediction interval.  

The conditional calibration of point forecasts and predictive distributions can be assessed using decompositions of scoring functions \citep{Gneiting2011} and scoring rules \citep{GneitingRaftery2007} respectively. While scoring functions and scoring rules summarise forecast accuracy in a single number, thus allowing competing forecasters to be ranked and compared, they can be decomposed into terms that measure different aspects of forecast performance, including forecast miscalibration. \cite{Sanders1963} and \cite{Murphy1973} first proposed decompositions of the squared error, which \cite{DeGrootFeinberg1983} and \cite{Brocker2009} generalised to any scoring rule. \cite{BentzienFriederichs2014} applied the decomposition to the quantile score (or pinball loss), and \cite{GneitingResin2023} introduced a general framework for decomposing scoring functions for specific functionals of a predictive distribution. A detailed history of score decompositions is provided by \cite{Mitchell2019}.

Several recent studies have proposed estimating score decomposition terms using isotonic regression: \cite{DimitriadisEtAl2021} originally proposed this when evaluating probability forecasts for binary outcomes; \cite{GneitingResin2023} extend this further to forecasts for particular functionals, which is applied to quantile forecasts by \cite{GneitingEtAl2023}; while \cite{ArnoldEtAl2023} employ isotonic distributional regression \citep[IDR;][]{HenziEtAl2021} to estimate the decomposition terms for the continuous ranked probability score \citep[CRPS;][]{MathesonWinkler1976} when evaluating the conditional calibration of full predictive distributions.

However, despite the importance of prediction intervals in many application domains, no score decompositions have been proposed to assess the conditional calibration of interval forecasts. In this work, we list several notions of interval forecast calibration, and then demonstrate how a strong notion of interval forecast calibration can be assessed via decompositions of the interval score (also referred to as the Winkler score), which is the canonical scoring function with which to assess interval forecasts \citep{Winkler1972,GneitingRaftery2007}. Following the work described above, by estimating the decomposition terms using isotonic distributional regression, we obtain a convenient framework with which to assess a strong conditional notion of the calibration of prediction intervals. We show that the resulting decomposition terms provide information that is not available from simply comparing the average length of unconditionally calibrated prediction intervals, which is commonly performed in practice \citep{ShaferVovk2008,FontanaEtAl2023}.

Notions of coverage and calibration of interval forecasts are introduced and compared in the following section, before Section \ref{sec:decomps} demonstrates how these can be assessed using decompositions of the interval score. These decompositions are applied to simulated data in Section \ref{sec:sims}, and baseline regression datasets in Section \ref{sec:apps}. All proofs are deferred to Appendix \ref{app:proofs}, while \verb!R! code to reproduce the results herein is available at \url{https://github.com/sallen12/CondIntCal_rep}.

\section{Calibration of interval forecasts}\label{sec:cali}

Suppose that we wish to forecast a real-valued outcome variable $Y \in \R$, and our forecast takes the form of a central $1 - \alpha$ prediction interval for some arbitrary fixed $\alpha \in (0, 1)$. We treat the interval forecast here as a random interval, with lower and upper bounds denoted by $L \in \R$ and $U \in \R$, respectively, with $L < U$ almost surely. In practice, forecast evaluation is performed over a set of many forecasts and observations, which correspond to realisations of $[L, U]$ and $Y$. We denote by $\Q$ the joint distribution of $L, U$ and $Y$, and, where relevant, assume that equalities involving these random variables hold almost surely.

The prediction interval $[L, U]$ may arise from the quantiles of a (random) predictive distribution $F$. We denote by $Q_{\alpha}(F)$ the set of $\alpha$-quantiles of $F$, and by $q_{\alpha}(F)$ an arbitrary element of this set. The set of quantiles $Q_{\alpha}(F)$ is either a singleton, making $q_{\alpha}(F)$ the unique $\alpha$-quantile of $F$, or a closed interval between the lower and upper $\alpha$-quantiles of $F$. We simplify the notation to $Q_{\alpha}(Y)$ and $q_{\alpha}(Y)$ to denote the $\alpha$-quantiles of the marginal distribution of $Y$, and $Q_{\alpha}(Y \mid X)$ and $q_{\alpha}(Y \mid X)$ for the $\alpha$-quantiles of the conditional distribution of $Y$ given a random variable $X$. 

The calibration of an interval forecast is typically assessed by checking whether the outcome falls within the interval with the desired coverage level,
\begin{equation}\label{eq:cal_unc}
	\Q(Y \in [L, U]) = 1 - \alpha.
\end{equation}
However, \eqref{eq:cal_unc} ignores that $[L, U]$ is a \emph{central} prediction interval, in that it neglects the probabilities with which the outcome falls above and below the interval; we additionally want that $\Q(Y < L) = \Q(Y > U) = \alpha/2$. We generalise this slightly by defining an interval forecast as (unconditionally) calibrated if it satisfies the following requirement.

\begin{defin}
    An interval forecast $[L, U]$ is \emph{unconditionally calibrated} if 
    \begin{equation}\label{eq:cal_unc2}
        \Q(Y < L) \le \frac{\alpha}{2} \le \Q(Y \le L), \quad \text{and} \quad \Q(Y > U) \le \frac{\alpha}{2} \le \Q(Y \ge U).
    \end{equation}
\end{defin}

If a forecast is unconditionally calibrated according to \eqref{eq:cal_unc2}, then 
\begin{equation}\label{eq:cal_unc1b}
    \Q(Y \in (L,U)) \le 1 - \alpha \le \Q(Y \in [L,U]),
\end{equation}
which simplifies to \eqref{eq:cal_unc} when $\Q(Y \in \{L, U\}) = 0$. This generalised notion of unconditional calibration accounts for discontinuities in the joint distribution of $(L,U,Y)$. In the literature, accommodating discontinuities in the joint distribution is often achieved by weakening \eqref{eq:cal_unc} to $\Q(Y \in [L, U]) \ge 1 - \alpha$, allowing conservative prediction intervals to be calibrated. The additional lower bound in \eqref{eq:cal_unc1b} strengthens the requirement for calibration by limiting how conservative a calibrated prediction interval can be. 

The definition of unconditional interval calibration at \eqref{eq:cal_unc2} is still fairly weak, since it averages across all possible forecasts and observations, neglecting possible conditional biases. One could argue that it therefore provides a minimum requirement for forecasts to be considered reliable, justifying checks for unconditional calibration in practice. However, it has repeatedly been remarked that forecasts should be conditionally calibrated, since this ensures that forecasts can be taken at face value in particular forecasting cases, not just on average \citep[see e.g.][]{Christoffersen1998,ShaferVovk2008,FontanaEtAl2023}. We adopt the following definition of conditional interval calibration.

\begin{defin}
	An interval forecast $[L, U]$ is \emph{auto-calibrated} if 
	\begin{equation}\label{eq:cal_con}
		L \in Q_{\frac{\alpha}{2}}(Y \mid L, U), \quad \text{and} \quad U \in Q_{1-\frac{\alpha}{2}}(Y \mid L, U).
	\end{equation}
\end{defin}

That is, $L$ and $U$ are conditional quantiles of $Y$ given the interval forecast itself. This can equivalently be written as 
\begin{equation}\label{eq:cal_con2}
    \Q(Y < L \mid L,U) \le \frac{\alpha}{2} \le \Q(Y \le L \mid L,U), \quad \text{and} \quad \Q(Y > U \mid L,U) \le \frac{\alpha}{2} \le \Q(Y \ge U \mid L,U),
\end{equation}
showing that it provides a conditional generalisation of \eqref{eq:cal_unc2}. If the conditional quantiles are unique, an interval forecast is auto-calibrated if and only if 
\begin{equation}\label{eq:cal_con3}
	\Q(Y < L \mid L, U) = \frac{\alpha}{2}, \quad \text{and} \quad \Q(Y > U \mid L, U) = \frac{\alpha}{2}.
\end{equation}
This in turn implies that the interval forecast has the correct conditional coverage, $\Q(Y \in [L, U] \mid L, U) = 1 - \alpha$, which further implies correct unconditional coverage at \eqref{eq:cal_unc}. 

This definition of conditional calibration is conditional on the interval forecast itself. This is akin to notions of auto-calibration of predictive distributions that are familiar in the forecast evaluation literature \citep{Tsyplakov2011,GneitingRanjan2013}. However, while desirable in theory, auto-calibration is generally difficult to assess in practice, since it is generally not possible to calculate the conditional distribution of $Y$ given the forecast \citep[see e.g.][]{CandilleTalagrand2005,ArnoldEtAl2023}. Similarly, auto-calibration of interval forecasts is difficult to assess in practice since it requires knowledge of the conditional quantiles $Q_{\alpha/2}(Y \mid L, U)$ and $Q_{1 - \alpha/2}(Y \mid L, U)$.

Rather than assessing conditional calibration with respect to all of the information provided by $[L, U]$, \cite{ArnoldEtAl2023} propose a notion of \emph{isotonic calibration} that evaluates forecasts with respect to a relevant subset of the information provided by the prediction interval. This leads to a conditional notion of forecast calibration that is slightly weaker than auto-calibration, but which facilitates an empirical assessment in practice. We follow a similar approach here for the calibration of interval forecasts. 

\begin{defin}
	An interval forecast $[L, U]$ is \emph{isotonically calibrated} if 
	\begin{equation}\label{eq:cal_iso}
		L \in Q_{\frac{\alpha}{2}}(Y \mid \Aa(L, U)) \quad \text{and} \quad U \in Q_{1-\frac{\alpha}{2}}(Y \mid \Aa(L, U)),
	\end{equation}
	where $\Aa(L, U)$ denotes the $\sigma$-lattice generated by $L$ and $U$. 
\end{defin}

That is, $L$ and $U$ are quantiles of the isotonic conditional law of $Y$ given the interval forecast \citep{ArnoldZiegel2023}. We refer readers to \cite{ArnoldZiegel2023} and references therein for further details on $\sigma$-lattices and isotonic conditional laws. While isotonic calibration does not correspond to a conditional coverage guarantee as at \eqref{eq:cal_con3}, under mild conditions, it implies that
\begin{equation}\label{eq:cal_qu}
    L \in Q_{\frac{\alpha}{2}}(Y \mid L) \quad \text{and} \quad U \in Q_{1-\frac{\alpha}{2}}(Y \mid U),
\end{equation}
see \citet[Lemma 6.8]{ArnoldZiegel2023}.
That is, the bounds $L$ and $U$ are quantile calibrated predictions for the $\alpha/2$ and $1 - \alpha/2$ quantiles of $Y$ respectively, in the sense of \cite{GneitingResin2023}. The requirement at \eqref{eq:cal_qu} could itself be used to define a notion of quantile calibration of interval forecasts, though since this is strictly weaker than isotonic calibration, we do not consider it in the following. Isotonic interval calibration additionally implies unconditional interval calibration, but is implied by auto-calibration. The relationships between these different notions of interval calibration are summarised below. Note that any conditional quantile $Z \in Q_\beta(Y \mid L,U)$ for some $\beta \in (0,1)$ can be written as  $Z = g(L,U)$ for some (measurable) function $g$ of $L$ and $U$.

\begin{propo}\label{prop:cal}
	The following relationships hold between the notions of interval calibration:
	\begin{itemize}
		\item[(i)] Auto-calibration implies unconditional calibration, but the converse is not true.
        \item[(ii)] Isotonic calibration implies unconditional calibration, but the converse is not true.
		\item[(iii)] Auto-calibration implies isotonic calibration.
        \item[(iv)] Isotonic calibration implies auto-calibration if, for any $g(L,U) \in Q_\beta(Y \mid L,U)$, $\beta \in \{\alpha/2,1-\alpha/2\}$, it holds that $g$ is componentwise increasing.
	\end{itemize}
\end{propo}

If $Y$ is continuous and the interval forecast is constant, then auto-calibration, isotonic calibration, and unconditional calibration are equivalent, since the interval $[L, U]$ is deterministic. More generally, auto-calibration is a stronger requirement than isotonic calibration, and both are much stronger than unconditional calibration. Isotonic and auto-calibration are equivalent if the conditional upper and lower quantiles of $Y$ are ordered (formalised in part (iv) above), which is a natural requirement for informative predictions, and should be satisfied at least approximately in many relevant forecasting situations.

Unconditional calibration is straightforward to assess in practice by calculating the proportion of observations that fall within the forecast intervals, whereas it is generally not possible to empirically assess the auto-calibration of interval forecasts. We demonstrate in the subsequent section how decompositions of the interval score can be used to assess isotonic calibration, following arguments in \cite{ArnoldEtAl2023} for predictive distributions.

\section{Interval score decompositions}\label{sec:decomps}

\subsection{The interval score}

Calibration is a necessary requirement for interval forecasts to be considered trustworthy. However, checks for calibration do not directly provide a means to compare competing forecasts. Instead, interval forecasts are typically compared using consistent scoring functions \citep{Gneiting2011}. A consistent scoring function takes a forecast and an observation as inputs, and outputs a numerical value that quantifies the forecast accuracy; consistent scoring functions are closely related to proper scoring rules for full predictive distributions with a lower score being indicative of a more accurate forecast \citep{GneitingRaftery2007,WaghmareZiegel2025}.

A scoring function is consistent for a functional of a distribution if the scoring function is minimised in expectation when the functional of the true outcome distribution is issued as the forecast. For example, a scoring function $S : \R^3 \to [0, \infty)$ is consistent for the central $1 - \alpha$ prediction interval if 
\begin{equation}\label{eq:cons}
	\E[S([\ell^*, u^*], Y)] \le \E[S([\ell, u], Y)]
\end{equation}
for all $\ell^* \in Q_{\alpha/2}(Y)$, $u^* \in Q_{1-\alpha/2}(Y)$, and $\ell, u \in \R$ with $\ell < u$. That is, the expected score is minimised when the interval forecast is derived from the relevant quantiles of the outcome $Y$. The scoring function is strictly consistent if equality holds in \eqref{eq:cons} if and only if $\ell \in Q_{\alpha/2}(Y)$ and $u \in Q_{1 - \alpha/2}(Y)$. Note that the interval forecast $[\ell, u]$ in \eqref{eq:cons} is not random.

Interval forecasts are typically assessed using the interval score, or Winkler score \citep{Winkler1972,GneitingRaftery2007},
\begin{align*}
    \mathrm{IS}_{\alpha}([\ell, u], y) &= | u - \ell | + \frac{2}{\alpha} \one\{y < \ell\} (\ell - y) + \frac{2}{\alpha} \one\{y > u\} (y - u) \\
    &= \frac{2}{\alpha} \left[ \mathrm{QS}_{\frac{\alpha}{2}}(\ell, y) + \mathrm{QS}_{1 - \frac{\alpha}{2}}(u, y) \right],
\end{align*}
where 
\begin{equation}\label{eq:qs}
    \mathrm{QS}_{\alpha}(x, y) = (\one\{y \le x\} - \alpha)(x - y)
\end{equation}
is the quantile score or pinball loss. The quantile score at level $\alpha$ is consistent for the $\alpha$-quantile, and, as a sum of quantile scores, the interval score is therefore consistent for the central $1 - \alpha$ prediction interval. 

The interval score is equal to the length of the interval forecast plus a penalty if the outcome falls outside the interval. The penalty increases as $\alpha$ decreases, recognising that the $1 - \alpha$ prediction interval should contain the outcome more often when $\alpha$ is close to zero. The interval score therefore implicitly rewards interval forecasts that are as short as possible, subject to containing the outcome with the desired coverage probability. This is made more explicit in the following section using decompositions of the interval score.

The interval score can be extended by applying non-decreasing transformations of $\ell, u$ and $y$, thereby changing how outcomes outside of the interval are penalised \citep{Gneiting2011}, and can also be generalised to non-central prediction intervals. These two extensions are considered in Appendix \ref{app:extensions}.

\subsection{Population level decomposition of the interval score}

While scoring functions condense forecast performance into a single numeric value, expected scores can be decomposed into terms quantifying different aspects of forecast performance \citep{Murphy1973,Brocker2009}. Most commonly, forecasts are decomposed into three terms, quantifying the inherent uncertainty in the forecast scenario, the discriminative ability of the forecast, and the forecast miscalibration. While \cite{BentzienFriederichs2014} and \cite{GneitingEtAl2023} both consider decompositions of quantile scores, decompositions of the interval score have not been studied.

For the interval score, we can write
\begin{equation}\label{eq:decomp_auto}
	\E[\text{IS}_{\alpha}([L, U], Y)] = \mathrm{UNC} - \mathrm{DSC}_{\mathrm{A}} + \mathrm{MCB}_{\mathrm{A}},
\end{equation}
where
\begin{align*}
	\mathrm{UNC} &= \E[\text{IS}_{\alpha}([q_{\frac{\alpha}{2}}(Y), q_{1 - \frac{\alpha}{2}}(Y)], Y)], \\
	\mathrm{DSC}_{\mathrm{A}} &= \E[\text{IS}_{\alpha}([q_{\frac{\alpha}{2}}(Y), q_{1 - \frac{\alpha}{2}}(Y)], Y)] - \E[\text{IS}_{\alpha}([q_{\frac{\alpha}{2}}(Y \mid L, U), q_{1 - \frac{\alpha}{2}}(Y \mid L, U)], Y)], \\
	\mathrm{MCB}_{\mathrm{A}} &= \E[\text{IS}_{\alpha}([L, U], Y)] - \E[\text{IS}_{\alpha}([q_{\frac{\alpha}{2}}(Y \mid L, U), q_{1 - \frac{\alpha}{2}}(Y \mid L, U)], Y)],
\end{align*}
and the expectation is taken over both the interval forecast $[L, U]$ and the outcome $Y$. These decomposition terms can be defined using arbitrary $\alpha/2$ and $1 - \alpha/2$ quantiles of $Y$ (i.e. any $q_{\alpha/2} \in Q_{\alpha/2}$ and $q_{1-\alpha/2} \in Q_{1-\alpha/2}$), since all result in the same expected interval score.

The first term, $\mathrm{UNC}$, is the expected interval score with the unconditional $\alpha/2$ and $1 - \alpha/2$ quantiles of $Y$ used to construct the prediction interval. This term is independent of the forecast and depends only on the outcome, providing a measure of \emph{uncertainty} in the forecasting scenario. The second term, $\mathrm{DSC}_{\mathrm{A}}$, is the expected improvement in score obtained by using the conditional quantiles of $Y$ given the prediction interval $[L, U]$ to define the interval forecast rather than the unconditional quantiles of $Y$, providing a measure of the \emph{discrimination ability} inherent in the forecast $[L, U]$. This term acts negatively on the score, so that more informative forecasts yield lower interval scores. The final term, $\mathrm{MCB}_{\mathrm{A}}$, is the expected improvement in score obtained by using the conditional quantiles of $Y$ given the prediction interval $[L, U]$ to define the interval forecast rather than $[L, U]$ itself, providing a measure of the \emph{miscalibration} in the forecast $[L, U]$. This term acts positively on the score, penalising forecasts that are miscalibrated.

Since the interval score is non-negative and strictly consistent for the central prediction interval, the decomposition at \eqref{eq:decomp_auto} satisfies several desirable properties.

\begin{propo}\label{prop:decomp_auto}
	The decomposition of the expected interval score at \eqref{eq:decomp_auto} satisfies the following properties:
	\begin{itemize}
		\item[(i)] $\mathrm{DSC}_{\mathrm{A}} \ge 0$ with equality if and only if $q_{\alpha/2}(Y \mid L, U) \in Q_{\alpha/2}(Y)$ and $q_{1 - \alpha/2}(Y \mid L, U) \in Q_{1 - \alpha/2}(Y)$.
		\item[(ii)] $\mathrm{MCB}_{\mathrm{A}} \ge 0$ with equality if and only if the interval forecast $[L, U]$ is auto-calibrated.
	\end{itemize}
\end{propo}

Of particular interest here is the miscalibration term, $\mathrm{MCB}_{\mathrm{A}}$. This term provides a measure of how far the interval score for $[L, U]$ differs from that of an auto-calibrated forecast interval using the same information, thereby quantifying the extent to which the forecast $[L, U]$ is not auto-calibrated. 

A similar decomposition can be used to evaluate isotonic interval calibration,
\begin{equation}\label{eq:decomp_iso}
	\E[\text{IS}_{\alpha}([L, U], Y)] = \mathrm{UNC} - \mathrm{DSC}_{\mathrm{I}} + \mathrm{MCB}_{\mathrm{I}},
\end{equation}
where $\mathrm{UNC}$ is as defined above, and
\begin{align*}
	\mathrm{DSC}_{\mathrm{I}} &= \E[\text{IS}_{\alpha}([q_{\frac{\alpha}{2}}(Y), q_{1 - \frac{\alpha}{2}}(Y)], Y)] - \E[\text{IS}_{\alpha}([q_{\frac{\alpha}{2}}(Y \mid \Aa(L, U)), q_{1 - \frac{\alpha}{2}}(Y \mid \Aa(L, U))], Y)], \\
	\mathrm{MCB}_{\mathrm{I}} &= \E[\text{IS}_{\alpha}([L, U], Y)] - \E[\text{IS}_{\alpha}([q_{\frac{\alpha}{2}}(Y \mid \Aa(L, U)), q_{1 - \frac{\alpha}{2}}(Y \mid \Aa(L, U))], Y)].
\end{align*}

This decomposition is similar to before, but with quantiles of the conditional distribution of $Y$ given $[L, U]$ replaced with quantiles of the isotonic conditional law of $Y$ given $[L, U]$. Using properties of the isotonic conditional law \citep[see][]{ArnoldZiegel2023}, we can deduce that this isotonic decomposition has analogous desirable properties as the decomposition at \eqref{eq:decomp_auto}.

\begin{propo}\label{prop:decomp_IDR_1}
	The decomposition of the expected interval score at \eqref{eq:decomp_iso} satisfies the following properties:
	\begin{itemize}
		\item[(i)] $\mathrm{DSC}_{\mathrm{I}} \ge 0$ with equality if and only if $q_{\alpha/2}(Y \mid \Aa(L, U)) \in Q_{\alpha/2}(Y)$ and $q_{1 - \alpha/2}(Y \mid \Aa(L, U)) \in Q_{1 - \alpha/2}(Y)$.
		\item[(ii)] $\mathrm{MCB}_{\mathrm{I}} \ge 0$ with equality if and only if the interval forecast $[L, U]$ is isotonically calibrated.
	\end{itemize}
\end{propo}

In this case, the miscalibration term $\mathrm{MCB}_{\mathrm{I}}$ measures the extent to which the interval forecast $[L, U]$ deviates from an isotonically calibrated forecast. Since auto-calibration implies isotonic calibration (Section \ref{sec:cali}), it follows that $\mathrm{MCB}_{\mathrm{I}}$ is also equal to zero when the forecast is auto-calibrated. More generally, the deviation from auto-calibration is always at least as large as the deviation from isotonic calibration, but these terms will often be the same since these two definitions are equivalent in many relevant forecasting situations (see Proposition \ref{prop:cal}).

\begin{propo}\label{prop:IDR_auto}
	It holds that 
	\[
	\E[\mathrm{IS}_{\alpha}([L, U], Y)] \ge \mathrm{MCB}_{\mathrm{A}} \ge \mathrm{MCB}_{\mathrm{I}}.
	\]
\end{propo}

\begin{remark}
	The expected interval score could also be decomposed to provide a measure of unconditional miscalibration, along the lines of \citet[Section 2.4]{Bashaykh2022}. However, since unconditional interval calibration can easily be assessed without such decompositions, we omit such an analysis here.
\end{remark}

\subsection{Sample level decomposition of the interval score}\label{sec:iso_dcmp_sample}

In practice, evaluation is performed over a set of $n$ interval forecasts $[\ell_{1}, u_{1}], \dots, [\ell_{n}, u_{n}]$ and corresponding outcomes $y_{1}, \dots, y_{n}$, which can be interpreted as realisations of the random interval $[L, U]$ and the outcome variable $Y$. To estimate the terms of the auto-calibration decomposition at \eqref{eq:decomp_auto}, we require estimates of the conditional quantiles of $Y$ given $[L, U]$. However, assuming the forecasts $[\ell_{1}, u_{1}], \dots, [\ell_{n}, u_{n}]$ are pairwise distinct, we only have one realisation of $Y$ given each observed interval forecast. Estimation of the conditional quantiles therefore requires making some assumptions about the relationship between $Y$ and $[L,U]$, which should ideally be as weak as possible.

\cite{ArnoldZiegel2023} demonstrate that the isotonic conditional law is the population version of isotonic distributional regression (IDR), a flexible non-parametric distributional regression procedure that assumes an isotonic relationship between the covariates and the outcome variable \citep{HenziEtAl2021}. IDR exhibits in-sample optimality properties with respect to a large class of proper scoring rules, including the continuous ranked probability score (CRPS). \cite{ArnoldEtAl2023} demonstrate that it can therefore be used within decompositions of the CRPS to measure isotonic calibration of predictive distributions. Since IDR is also optimal with respect to the interval score, a similar framework can be applied to decompositions of the interval score, facilitating an assessment of the isotonic calibration of interval forecasts. 

To estimate the isotonic decomposition terms at \eqref{eq:decomp_iso}, we apply IDR to the observations $y_1, \dots, y_n$ using the interval forecasts $[\ell_1, u_1], \dots, [\ell_n, u_n]$ as covariates. IDR requires a partial order on the space of covariates, and we say that $[\ell_i, u_i] \le [\ell_j, u_j]$ if $\ell_i \le \ell_j$ and $u_i \le u_j$. That is, an interval forecast is larger (smaller) than another if the bounds of the interval are larger (smaller) than those of the other forecast; in this case, we say that the two forecasts are \emph{comparable} (if they are not comparable, the intervals are nested). This ordering is implied if the interval forecasts are derived from distributions that are stochastically ordered \citep{ArnoldEtAl2023}, and is closely related to the requirement given in Proposition \ref{prop:cal}(iv) for isotonic calibration and auto-calibration to be equivalent. In theory, other orderings on the space of interval forecasts could also be considered, though not all orderings will result in meaningful decomposition terms; this is discussed further in Appendix \ref{app:orderings}. 

Applying IDR to the interval forecasts and observations yields discrete predictive distributions $\tilde{F}_1, \dots, \tilde{F}_n$, and we denote by $\tilde{Q}_{\alpha, i}$ the set of $\alpha$-quantiles of $\tilde{F}_i$, with $\tilde{q}_{\alpha, i}=\min \tilde{Q}_{\alpha,i}$ being the lower quantile of $\tilde{F}_i$, for $i = 1, \dots, n$. We similarly denote by $\bar{Q}_{\alpha}$ the set of $\alpha$-quantiles of $y_{1}, \dots, y_{n}$, and $\bar{q}_{\alpha}=\min\bar{Q}_{\alpha}$. The general idea is that IDR provides a recalibrated interval forecast, and forecasts are then assessed relative to this recalibrated prediction. In particular, we approximate the terms of \eqref{eq:decomp_iso} using the sample decomposition
\begin{equation}\label{eq:decomp_IDR}
	\widehat{\mathrm{IS}}=\frac{1}{n} \sum_{i=1}^{n} \text{IS}_{\alpha}([\ell_i, u_i], y_i) = \widehat{\mathrm{UNC}} - \widehat{\mathrm{DSC}}_{\mathrm{I}} + \widehat{\mathrm{MCB}}_{\mathrm{I}},
\end{equation}
where
\begin{align}
	\widehat{\mathrm{UNC}} &= \frac{1}{n} \sum_{i=1}^{n} \text{IS}_{\alpha}([\bar{q}_{\frac{\alpha}{2}}, \bar{q}_{1 -\frac{\alpha}{2}}], y_{i}), \\
	\widehat{\mathrm{DSC}}_{\mathrm{I}} &= \frac{1}{n} \sum_{i=1}^{n} \text{IS}_{\alpha}([\bar{q}_{\frac{\alpha}{2}}, \bar{q}_{1 -\frac{\alpha}{2}}], y_{i}) - \frac{1}{n} \sum_{i=1}^{n} \text{IS}_{\alpha}([\tilde{q}_{\frac{\alpha}{2}, i}, \tilde{q}_{1 -\frac{\alpha}{2}, i}], y_{i}),\label{eq:hatDSC} \\
	\widehat{\mathrm{MCB}}_{\mathrm{I}} &= \frac{1}{n} \sum_{i=1}^{n} \text{IS}_{\alpha}([\ell_i, u_i], y_{i}) - \frac{1}{n} \sum_{i=1}^{n} \text{IS}_{\alpha}([\tilde{q}_{\frac{\alpha}{2}, i}, \tilde{q}_{1 -\frac{\alpha}{2}, i}], y_{i}). \label{eq:hatMCB}
\end{align}

Each of these terms corresponds to a sample estimate of the respective terms in \eqref{eq:decomp_iso}. This decomposition is exact, in that the terms of the decomposition cancel out to recover the average interval score of the forecasts $[\ell_1, u_1], \dots, [\ell_n, u_n]$. Since the interval score is non-negative, it follows that $\widehat{\mathrm{UNC}} \ge 0$ and this term is independent of the forecasts as in the population decomposition. Moreover, due to the in-sample optimality of IDR with respect to the interval score, this sample decomposition also exhibits the following desirable properties.

\begin{propo}\label{prop:decomp_IDR}
	The sample decomposition of the interval score at \eqref{eq:decomp_IDR} has the following properties:
	\begin{itemize}
		\item[(i)] $\widehat{\mathrm{DSC}}_{\mathrm{I}} \ge 0$ with equality if and only if $\bar{q}_{\alpha/2} = \tilde{q}_{\alpha/2,i} $ and $\bar{q}_{\alpha/2} =\tilde{q}_{1-\alpha/2,i}$ for all $i = 1, \dots, n$.
		\item[(ii)] $\widehat{\mathrm{MCB}}_{\mathrm{I}} \ge 0$ with equality if and only if $\ell_i \in \tilde{Q}_{\alpha/2, i}$ and $u_i \in \tilde{Q}_{1-\alpha/2, i}$ for all $i = 1, \dots, n$ and $(\ell_1,\dots,\ell_n)$, $(u_1,\dots,u_n)$ are valid solutions to the respective isotonic quantile regression problems.
	\end{itemize}
\end{propo}

\begin{remark}
    The last condition that $(\ell_1,\dots,\ell_n)$, $(u_1,\dots,u_n)$ are valid solutions to the respective isotonic quantile regression problems means that $(\ell_1,\dots,\ell_n)$ needs to be a minimiser of the criterion
    \begin{equation}\label{eq:quantile_average}
    \sum_{i=1}^n \mathrm{QS}_{\alpha/2}(x_i,y_i)
    \end{equation}
    over all sequences $(x_1, \dots, x_n)$ such that $x_i\leq x_j$ if $[l_i,u_i] \leq [l_j,u_j]$, where the quantile score is defined at \eqref{eq:qs}, and analogously for $(u_1,\dots,u_n)$. Any such solution will satisfy $\ell_i \in \tilde{Q}_{\alpha/2, i}$ for all $i=1,\dots,n$. Conversely, while it is tempting to think that all  sequences that satisfy $\ell_i \in \tilde{Q}_{\alpha/2, i}$, $i=1,\dots,n$ and the ordering constraints, are valid solutions, this is unfortunately not true, see \citet[Section 4.1]{JordanEtAl2022}. However, the sequence of lower quantiles, upper quantiles or any convex combination thereof is always a valid solution.
\end{remark}

\begin{remark}\label{rem:ass}
	The properties listed in Proposition \ref{prop:decomp_IDR} hold without any requirements on the data. However, for the terms of the sample decomposition to be practically meaningful, three things are required:
	\begin{itemize}
		\item Sufficient data is available: IDR typically requires an (effective) sample of size of at least $n = 500$ to provide reasonable estimates of the decomposition terms.
		\item There is an (approximately) isotonic relationship between the interval forecasts $[\ell_1, u_1], \dots, [\ell_n, u_n]$ and the observations $y_1, \dots, y_n$. This assumes that the original interval forecasts are already reasonably good. This is not necessary for the theoretical properties of the decomposition terms derived in this section to hold, but the notion of isotonic calibration is more informative when there is a stronger isotonic relationship; essentially the conditioning is then based on more information, as highlighted by Proposition \ref{prop:cal}(iv).
		\item The interval forecasts $[\ell_1, u_1], \dots, [\ell_n, u_n]$ are (mostly) comparable, in that $[\ell_i, u_i] \le [\ell_j, u_j]$ or $[\ell_i, u_i] \ge [\ell_j, u_j]$ holds for a large proportion of the data. Incomparable pairs of forecasts reduce the effective sample size used to train IDR, which can violate the first requirement above.
	\end{itemize}
\end{remark}

We argue that these three properties are satisfied in many practical applications. If the forecasts are not comparable, or if a larger interval forecast does not imply a larger observation, then an assessment of conditional calibration is somewhat unnecessary, since the interval forecasts are clearly not informative. In some applications, the required sample size to fit IDR will be prohibitive, though we demonstrate in the following sections the practical utility of the isotonic decomposition of the interval score when sufficient data is available. IDR can also be computationally expensive if it is fit to large datasets, in which case it may be useful to combine IDR with subsample aggregation (subbagging) strategies designed to reduce computational cost \citep[Section 2.3]{HenziEtAl2021}.

\section{Simulation examples}\label{sec:sims}

Following, \cite{GneitingEtAl2007}, suppose $Y \mid \mu \sim N(\mu, 1)$, where $\mu \sim N(0, 1)$, and let $\tau$ be equal to $1$ and $-1$ each with probability 1/2, independently of $(\mu, Y)$. Consider the following forecast distributions:
\bigskip

\(
\begin{alignedat}{2}
  &\bullet\quad \text{The \emph{unconditional} forecast:} &\quad& N(0, 2) \\
  &\bullet\quad \text{The \emph{ideal} forecast:} &\quad& N(\mu, 1) \\
  &\bullet\quad \text{The \emph{unfocused} forecast:} &\quad& (1/2) N(\mu,1) + (1/2) N(\mu+\tau,1) \\
  &\bullet\quad \text{The \emph{mean-biased} forecast:} &\quad& N(\mu+\tau,1) \\
  &\bullet\quad \text{The \emph{sign-biased} forecast:} &\quad& N(-\mu,1) \\
  &\bullet\quad \text{The \emph{mixed} forecast:} &\quad& \one\{\tau = 1\}N(0,2) + \one\{\tau = -1\}N(-\mu,1)
\end{alignedat}
\)

\bigskip
We generated central 0.9-level prediction intervals from these six forecasts, and evaluated their calibration using the decomposition of the interval score proposed in the previous section. The unconditional and ideal prediction intervals are auto-calibrated, and therefore also isotonically and unconditionally calibrated. However, while the unconditional forecast contains no more information than the marginal distribution of the response $Y$, the ideal forecast has knowledge of $\mu$, and should therefore be more discriminative. The unfocused prediction intervals are unconditionally calibrated, but neither isotonic nor auto-calibrated. The mean-biased, sign-biased, and mixed forecasts satisfy no notion of calibration considered here, but each contain different amounts of information.

We simulated $n=1000$ realisations of $\mu, \tau,$ and $Y$, and evaluated the calibration of the corresponding prediction intervals. For each method, IDR was applied to the 1000 prediction intervals using the \verb!isodistrreg! package in \verb!R! \citep{isodistrreg}, from which the (lower) $\alpha/2$ and $1-\alpha/2$ quantiles were obtained in order to calculate the interval score decomposition terms; while we chose $\tilde{q}_{\alpha/2,i}$ and $\tilde{q}_{1-\alpha/2,i}$ to be the lower $\alpha/2$ and $1-\alpha/2$ quantiles of the corresponding IDR predictive distribution respectively, this choice is somewhat arbitrary and should not significantly affect the resulting decomposition terms, especially when $n$ is large. The decomposition terms for the six forecasts are displayed in Figure \ref{fig:ss_mcb_dsc}. Miscalibration-discrimination plots display the discrimination term of the interval score decomposition against the miscalibration term, with parallel grey isolines corresponding to different mean interval scores \citep{Dimitriadis_2023}, and the uncertainty term of the decomposition shown in green. Methods that issue more accurate forecasts, as assessed using the interval score, will be closer to the top left corner of the plot.

\begin{figure}[t!]
\centering
 \includegraphics[width=0.6\linewidth]{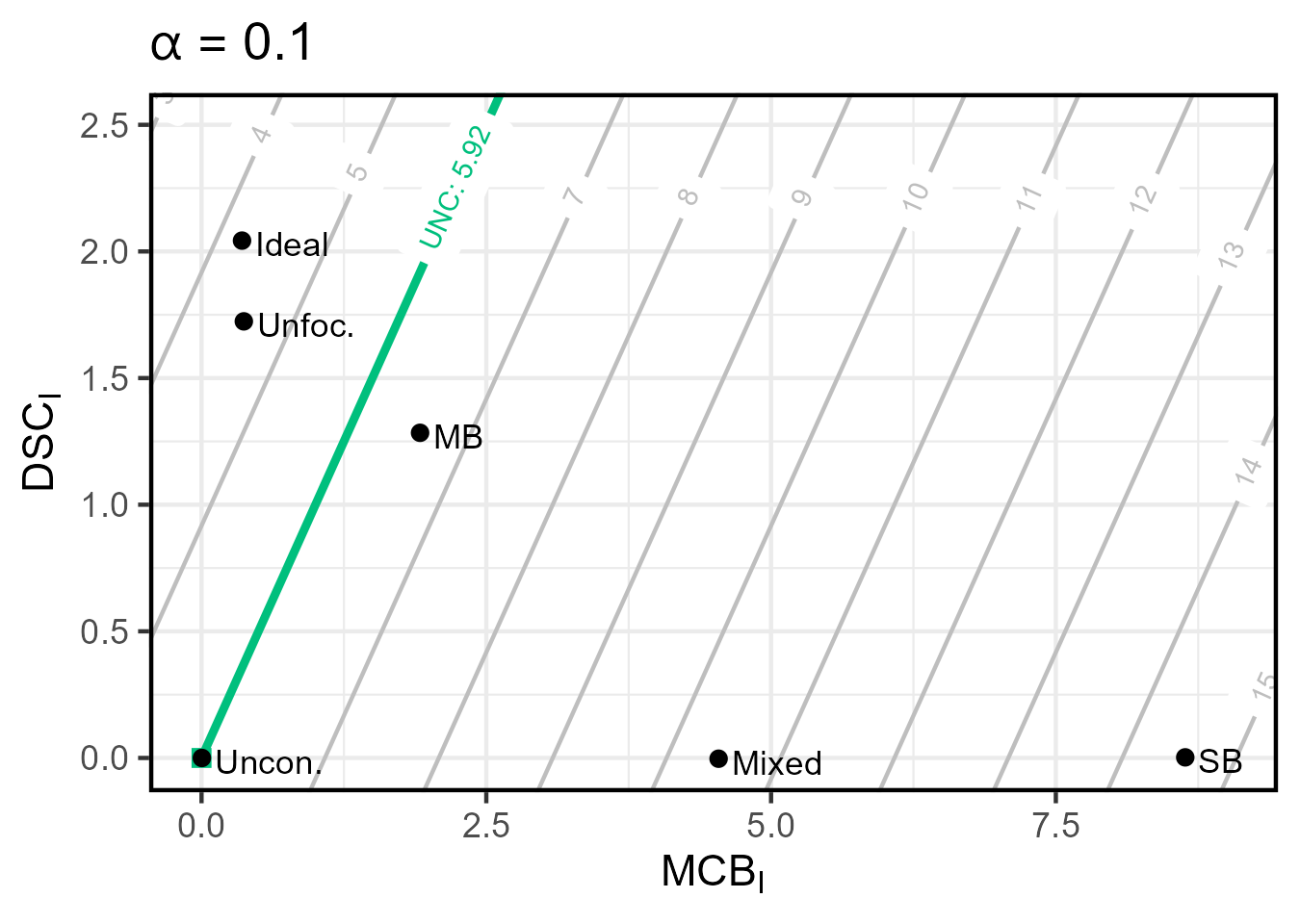}
  \caption{Miscalibration-discrimination plot for the ideal, the unconditional, the unfocused, the mean-biased (MB), the sign-based (SB) and the mixed forecaster with $\alpha=0.1$ and $n=1000$. The miscalibration and discrimination terms of the isotonic interval score decomposition at \eqref{eq:decomp_iso} are shown on the x- and y-axis, respectively. The diagonal grey lines are isolines at different values of the interval score, with the green line indicating the uncertainty term of the decomposition, when $\text{MCB}=\text{DSC}=0$.}
  \label{fig:ss_mcb_dsc}
\end{figure}

As expected, the climatological forecast is perfectly calibrated yet completely uninformative with $\text{MCB}_{\text{I}} = \text{DSC}_{\text{I}} = 0$. In this case, the forecast is constant, yielding constant IDR distributions $\Tilde{F}_1=\dots=\Tilde{F}_n$ with quantiles that correspond to the empirical quantiles of the realised observations. The ideal forecast is also calibrated, but considerably more discriminative, resulting in a lower average interval score. The fact that $\text{MCB}_{\text{I}}$ is not exactly zero for the ideal forecast is due to the finite sample size used to calculate the decomposition terms; the miscalibration term tends towards zero as the sample size increases. The mean-biased forecast receives a higher average score since it performs worse than the ideal forecast with respect to both miscalibration and discrimination. As a mixture between the ideal and the mean-biased forecasts, the unfocused forecast produces miscalibration and discrimination terms between those of its two component forecasts.

Although the sign-biased forecast contains the information provided by $\mu$, it receives a discrimination term of zero. The reason for this is that there is an antitonic relationship between the prediction intervals and the outcome, meaning the forecast contains no isotonic information, which is what is measured by the isotonic discrimination term. One could argue that this is a negative aspect of the decomposition proposed herein (and of score decompositions based on isotonic regression more generally). However, an antitonic relationship between the forecasts and outcomes is unlikely to manifest in practical applications, and one could argue that the conditional calibration of such forecasts is anyway irrelevant. Since the IDR recalibrated prediction intervals follow the same ordering as the original interval forecasts by design, the sign-biased interval forecasts are additionally strongly miscalibrated, resulting in the worst average interval score. The mixed forecast is equal to the climatological and sign-biased forecasts with equal probability, and therefore also has zero discrimination ability, as well as a miscalibration term halfway between those of the two forecasts. 

Table \ref{tab:ss_cov} displays the average interval score for the six forecasts, as well as the unconditional coverage and average length of the prediction intervals before and after being recalibrated using IDR. For the recalibrated prediction intervals, we show the coverage of the open and closed interval forecasts, i.e. $\Q(Y \in (L, U))$ and $\Q(Y \in [L, U])$. Since IDR yields prediction intervals that are unconditionally calibrated, the nominal coverage level $1 - \alpha$ is situated between these two coverages (see \eqref{eq:cal_unc1b}). The open and closed coverages are the same for the prediction intervals before recalibration, since $\Q(Y \in \{L, U\}) = 0$ for the original interval forecasts.

The average length of the interval forecasts can either increase or decrease after the recalibration procedure, as only the order of the intervals is relevant and used to fit IDR. Hence, to determine whether the isotonic decomposition could reasonably be applied to a collection of interval forecasts in practice, one could count how many pairs of prediction intervals are ordered. If only a small number of prediction intervals are ordered, then the resulting IDR distributions will be coarse step functions, with sets of quantiles that correspond to relatively large intervals. In this case, the difference between $\Q(Y \in (L, U))$ and $\Q(Y \in [L, U])$ can be large, making the notion of calibration at \eqref{eq:cal_unc1b} somewhat weaker. 

\begin{table}[t!]
    \centering
    \begin{tabular}{l|ccccccc}
        \multicolumn{2}{c}{} & \multicolumn{2}{c}{Original} & \multicolumn{2}{c}{Recalibrated} \\
        & Interval score & Coverage & Length & Coverage & Length\\
        \hline
        \hline
        Ideal & 4.23 & 0.88 & 3.29 & 0.87 - 0.92 & 3.29 \\
        Climatological & 5.92 & 0.91 & 4.65 & 0.90 - 0.90 & 4.56 \\
        Unfocused & 4.56 & 0.89 & 3.68 & 0.87 - 0.92 & 3.57 \\
        Mean-biased & 6.55 & 0.73 & 3.29 & 0.87 - 0.92 & 3.86 \\
        Sign-biased & 14.55 & 0.56 & 3.29 & 0.90 - 0.90 & 4.56 \\
        Mixed & 10.46 & 0.73 & 3.95 & 0.90 - 0.90 & 4.56 \\
    \end{tabular}
    \caption{Coverage and length of the 0.9-level prediction intervals from the six forecasts in the simulation study, before and after recalibration using IDR. Average interval scores of the original forecasts are also shown. For the IDR-recalibrated intervals, the coverage is shown for both the open $(L, U)$ and closed $[L, U]$ intervals; the unconditional calibration of IDR ensures that 0.9 lies between these two coverages (see \eqref{eq:cal_unc1b}).}
    \label{tab:ss_cov}
\end{table}

\section{Applications}\label{sec:apps}

\subsection{Data}

In this section, we apply the interval score decomposition to prediction intervals obtained from several prediction algorithms applied to three baseline regression datasets. The data and methods are a subset of those considered by \cite{RomanoEtAl2019}. While \cite{RomanoEtAl2019} study the average length and coverage of prediction intervals across 11 benchmark datasets, we present results separately for three of these datasets of different sizes, allowing us to analyse how the amount of data affects the interval score decomposition terms. We consider predictions for the Student-Teacher Achievement Ratio (STAR) in Tennessee \citep{Mosteller1995}, for which $n=433$ test data points are available; the usage of a bike sharing system in Washington D.C. \citep{FanaeeTGama2014} with $n=2178$ test data points; and the number of comments that a Facebook post will receive \citep{facebook} for which $n=8190$.

For each of the three datasets, prediction intervals were obtained from eight different methods. These methods include three versions of the original split conformal prediction algorithm \citep{VovkEtAl2018} using random forests (RF), ridge regression (Ridge), and neural networks (Net), as well as locally adaptive variants of each. The local variants apply the conformal prediction algorithms to normalised residuals, allowing the methods to adapt to heteroskedasticity \citep{PapadopoulosEtAl2008,LeiEtAl2018}. We additionally consider the conformalised quantile regression approach of \cite{RomanoEtAl2019} in combination with a neural network (CQR\_Net), as well as a non-conformal quantile regression method combined with a neural network (Q\_Net). Additional details regarding the methods can be found in the supplementary material of \cite{RomanoEtAl2019}, while implementation of the methods was performed using code from \url{https://github.com/yromano/cqr}. Two additional methods based on quantile random forests were also considered in \cite{RomanoEtAl2019}, though these are omitted from the analysis here since no working code is available to reproduce them. Implementation of IDR was performed using the \verb!isodistrreg! package in \verb!R! and all prediction intervals are generated at level $\alpha=0.1$. 

\subsection{Results}

Table \ref{tab:cs_cov} presents the average interval score, unconditional coverage, and average length of the prediction intervals obtained from the eight methods for each of the three datasets. Every method has approximately the correct unconditional coverage in all datasets. The different methods yield intervals of very similar lengths for the STAR dataset, resulting in similar interval scores, whereas the lengths and scores vary more between the methods in the other two datasets. The CQR and quantile neural networks always yield the smallest intervals and lowest interval scores.

\begin{table}[b!]
    \centering
    \begin{tabular}{l|cccccc}
        \multicolumn{3}{c}{} & \multicolumn{2}{c}{Original} & \multicolumn{2}{c}{Recalibrated} \\
        STAR & Comp. (\%) & Interval score & Coverage & Length & Coverage & Length \\
        \hline
        Ridge & 100.0 & 0.22 & 0.87 & 0.17 & 0.87 - 0.91 & 0.17 \\
        Local Ridge & 51.0 & 0.23 & 0.89 & 0.19 & 0.75 - 0.94 & 0.15 \\
        Random Forests & 100.0 & 0.22 & 0.89 & 0.18 & 0.88 - 0.92 & 0.17 \\
        Local Random Forests & 99.1 & 0.22 & 0.89 & 0.18 & 0.88 - 0.92 & 0.17 \\
        Neural Net & 100.0 & 0.25 & 0.91 & 0.20 & 0.88 - 0.91 & 0.19 \\
        Local Neural Net & 97.6 & 0.28 & 0.89 & 0.22 & 0.86 - 0.92 & 0.18 \\
        CQR Neural Net & 91.0 & 0.26 & 0.88 & 0.19 & 0.87 - 0.92 & 0.18 \\
        Quantile Neural Net & 88.7 & 0.23 & 0.91 & 0.20 & 0.84 - 0.93 & 0.17 \\
    \end{tabular}
    \bigskip
    
    \begin{tabular}{l|cccccc}
        \multicolumn{3}{c}{} & \multicolumn{2}{c}{Original} & \multicolumn{2}{c}{Recalibrated} \\
        Bike & Comp. (\%) & Interval score & Coverage & Length & Coverage & Length \\
        \hline
        Ridge & 100.0 & 3.30 & 0.89 & 2.21 & 0.88 - 0.92 & 2.11 \\
        Local Ridge & 60.7 & 2.99 & 0.89 & 2.18 & 0.81 - 0.93 & 1.87 \\
        Random Forests & 100.0 & 1.19 & 0.89 & 0.71 & 0.85 - 0.93 & 0.60 \\
        Local Random Forests & 99.3 & 1.16 & 0.88 & 0.68 & 0.84 - 0.93 & 0.59 \\
        Neural Net & 100.0 & 1.30 & 0.89 & 0.73 & 0.86 - 0.93 & 0.67 \\
        Local Neural Net & 99.1 & 1.15 & 0.88 & 0.65 & 0.85 - 0.93 & 0.60 \\
        CQR Neural Net & 96.1 & 0.93 & 0.90 & 0.59 & 0.82 - 0.93 & 0.58 \\
        Quantile Neural Net & 94.5 & 0.79 & 0.90 & 0.59 & 0.80 - 0.94 & 0.49 \\
    \end{tabular}
    \bigskip
    
    \begin{tabular}{l|cccccc}
        \multicolumn{3}{c}{} & \multicolumn{2}{c}{Original} & \multicolumn{2}{c}{Recalibrated} \\
        Facebook & Comp. (\%) & Interval score & Coverage & Length & Coverage & Length \\
        \hline
        Ridge & 100.0 & 12.59 & 0.90 & 3.72 & 0.37 - 0.95 & 3.05 \\
        Local Ridge & 59.5 & 8.41 & 0.90 & 3.46 & 0.35 - 0.96 & 2.86 \\
        Random Forests & 100.0 & 9.26 & 0.91 & 1.87 & 0.36 - 0.96 & 2.09 \\
        Local Random Forests & 96.9 & 7.56 & 0.90 & 1.62 & 0.36 - 0.96 & 2.01 \\
        Neural Net & 100.0 & 11.52 & 0.90 & 2.34 & 0.36 - 0.96 & 2.72 \\
        Local Neural Net & 85.3 & 8.39 & 0.90 & 1.78 & 0.35 - 0.96 & 2.58 \\
        CQR Neural Net & 59.9 & 6.17 & 0.90 & 1.11 & 0.35 - 0.96 & 2.12 \\
        Quantile Neural Net & 58.5 & 4.34 & 0.89 & 1.57 & 0.35 - 0.97 & 1.94 \\
    \end{tabular}
    \caption{Coverage and length of the 0.9-level prediction intervals from the eight prediction methods applied to the STAR (top), Bike (middle), and Facebook (bottom) datasets. The proportion of comparable interval forecasts for each prediction method in each dataset are also shown (Comp (\%)), as well as the average interval scores of the original forecasts.}
    \label{tab:cs_cov}
\end{table}

The length and unconditional coverage of the IDR recalibrated prediction intervals are also shown in Table \ref{tab:cs_cov}. Due to in-sample calibration guarantees of IDR, the coverage inequality at \eqref{eq:cal_unc1b} is satisfied by construction. To illustrate this, the coverages of the open and closed prediction intervals are shown, which always contain $1 - \alpha$. A small difference between these open and closed coverages suggests a smoother IDR fit, resulting in smaller sets of quantiles and more reliable estimates of the interval score decomposition terms. The Local Ridge approach tends to have a larger difference between these coverages since it often produces nested interval forecasts, reducing the effective sample size to fit IDR (see Remark \ref{rem:ass}). 

There is an even larger difference between these coverages for the Facebook dataset. This is because the variable of interest is censored below at zero, resulting in a high probability that $Y = L$ when the lower bound $L = 0$, in which case $Y \notin (L, U)$ but $Y \in [L, U]$. This is generally not an issue in practice since central predictions are of less interest when considering bounded or censored variables as a one-sided prediction intervals are generally preferred here. Nonetheless, we consider a central prediction interval to retain a comparison with the results of \cite{RomanoEtAl2019}.

Figure \ref{fig:cs_mcb_dsc} presents miscalibration-discrimination plots obtained from the isotonic decomposition of the interval score. In the STAR dataset, the miscalibration term of most methods is larger than the discrimination term, meaning they perform worse than a hypothetical unconditional forecast, whose score is equal to the uncertainty term of the decomposition. The reason for this is likely due to the smaller amount of data used to train the models, meaning the simpler ridge regression models outperform the more complex neural network-based methods. The high discrimination and miscalibration of the Local Ridge method is likely due to the comparatively low percentage of comparable interval forecasts (see Table \ref{tab:cs_cov}).

The forecasting methods are more useful when applied to the Bike and Facebook datasets, where they are ranked similarly: the best methods overall are the quantile and CQR neural networks, followed in general by the localised and then non-localised versions of the random forest and neural network methods, with the ridge regression forecasts performing worse. The ridge-based methods are less informative than the other methods, resulting in a lower discrimination term of the decomposition. The other methods all attain a similar discrimination term, but the flexibility of the quantile and CQR neural networks generally allows them to produce better calibrated forecasts, particularly when the sample size $n$ is large. However, the miscalibration term is non-zero for all methods, suggesting none of the generated forecasts are isotonically calibrated (and therefore also not auto-calibrated).

The decomposition terms in the miscalibration-decomposition plots provide information that is not available by simply comparing the average lengths of unconditionally calibrated prediction intervals, which is a commonly used method for comparing competing interval forecasts. For example, in the Facebook dataset, the local random forests generate relatively short prediction intervals while achieving a perfect unconditional coverage of 0.9. However, Figure \ref{fig:cs_mcb_dsc} and Table \ref{tab:cs_cov} demonstrate that the shorter intervals are indicative of more informative forecasts (seen by a large discrimination term), but not of better conditional calibration (seen by a large miscalibration term). The local random forests achieve a worse overall interval score than other methods, despite their shorter intervals and perfect unconditional coverage. 

\begin{figure}
    \centering
    \includegraphics[width=0.32\linewidth]{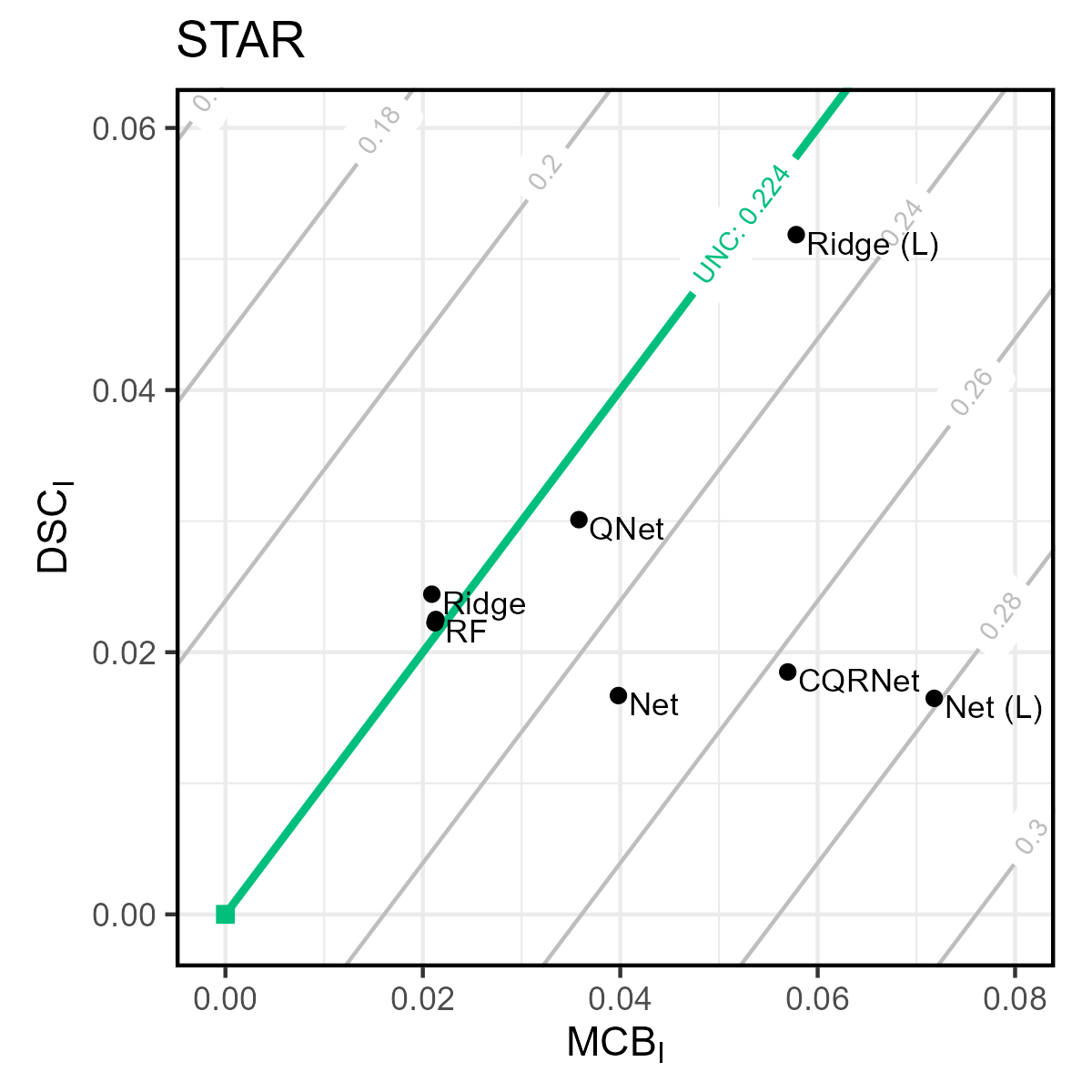}
    \includegraphics[width=0.32\linewidth]{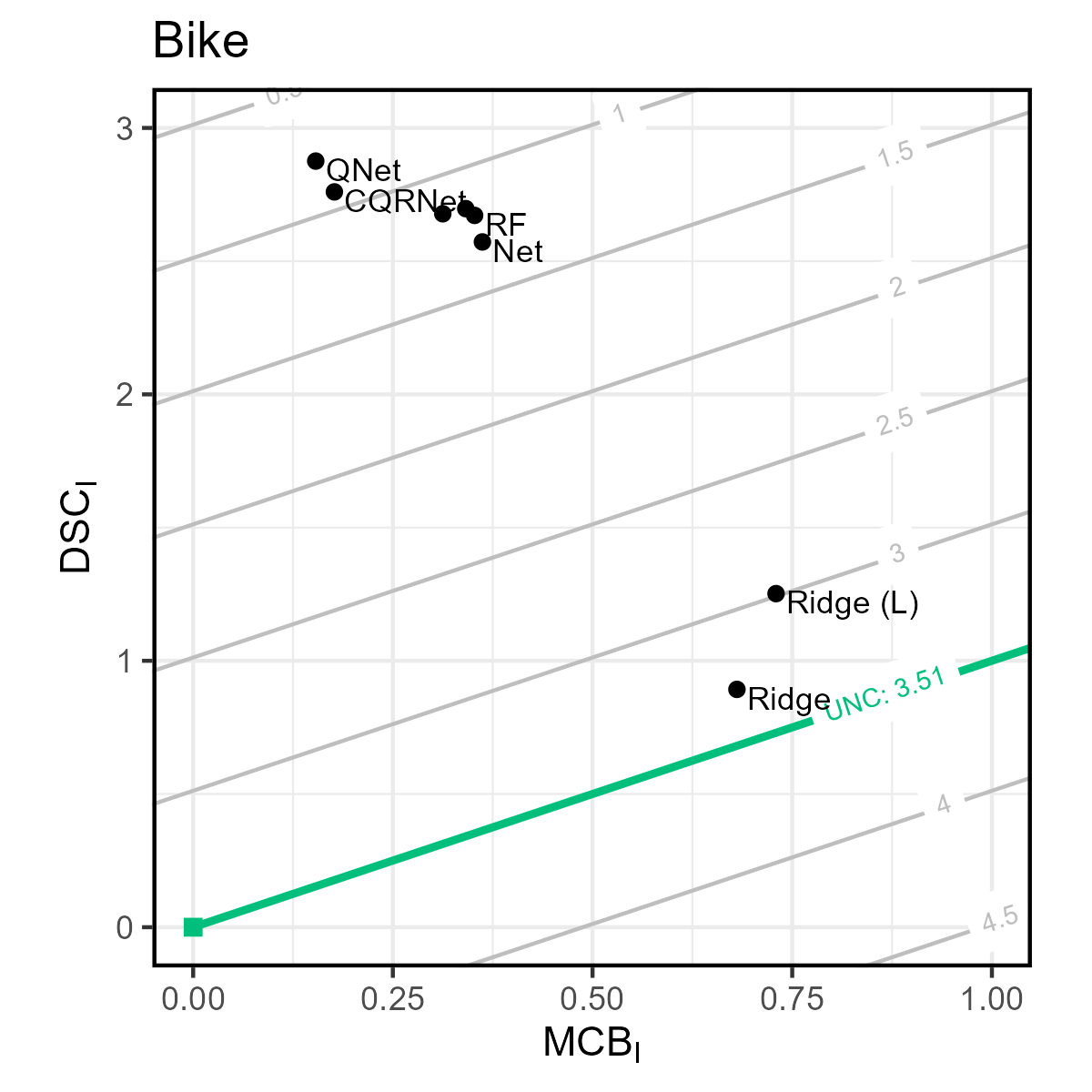}
    \includegraphics[width=0.32\linewidth]{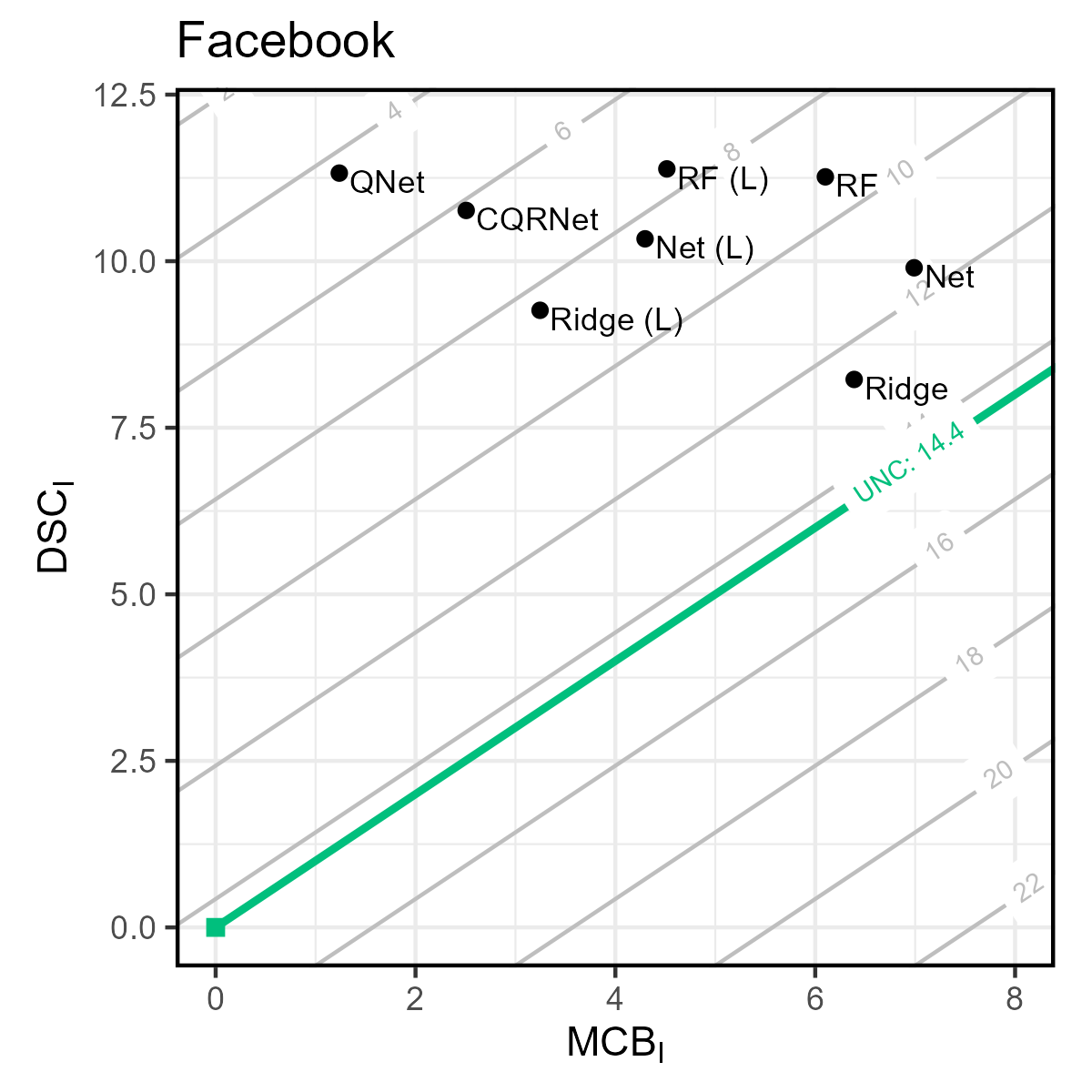}
    \caption{Isotonic miscalibration-discrimination plots for the eight prediction methods and three datasets. See Figure \ref{fig:ss_mcb_dsc} for details. For the STAR dataset, the decomposition terms of RF and RF (L) are equivalent. For the Bike dataset, the RF, Net, RF (L), and Net (L) are clustered together. Some of these labels have therefore been omitted for ease of presentation.}
    \label{fig:cs_mcb_dsc}
\end{figure}

\section{Conclusion}\label{sec:conc}

This paper provides a general framework with which to assess the conditional calibration of interval forecasts. The framework leverages a decomposition of the interval score into terms quantifying the discrimination ability and miscalibration of the forecast as well as the inherent uncertainty of the forecasting problem. By estimating the decomposition terms using isotonic distributional regression, we obtain a meaningful and interpretable decomposition that assesses a strong notion of forecast calibration while being easy to implement in practice. Following \cite{Dimitriadis_2023}, the decomposition terms can additionally be used to generate informative miscalibration-discrimination diagnostic plots that simultaneously display the discrimination ability and the (conditional) miscalibration of interval forecasts.

The framework is more theoretically appealing than simply comparing the average length of prediction intervals that are unconditionally calibrated, which is the common procedure in conformal prediction, for example. If interval forecasts are auto-calibrated, then a shorter interval corresponds to a more informative prediction that will receive a better score when assessed using consistent scoring functions. However, the same is not true when the forecasts are only unconditionally calibrated, in which case no direct connection between the length of the prediction interval and conditional calibration or forecast accuracy can be established.

For reliable empirical results, the interval score decomposition assumes that the forecasts to be assessed are (roughly) isotonic, in that a higher interval generally corresponds to a higher outcome. The requirement that interval forecasts are ordered is easier to satisfy than requiring that entire forecast distributions are stochastically ordered (the former is necessary for the latter), making the decomposition of the interval score easier to implement in practice than the more general decomposition of the CRPS proposed by \cite{ArnoldEtAl2023} for full predictive distributions. The approach additionally requires a moderately large amount of forecasts and observations in the evaluation data on which to fit the isotonic regression models. While this may limit the applicability of the approach in some applications, we demonstrate in Section \ref{sec:apps} how it can be successfully applied when sufficient data is available.

While we focus here on central prediction intervals, the methods can readily be extended to generalised interval scores, and to non-central prediction intervals; these extensions are detailed in Appendix \ref{app:extensions}. Moreover, the approach could also be used to evaluate the miscalibration of interval forecasts at several coverage levels simultaneously. For example, epidemiological forecasts typically take the form of collections of many prediction intervals, and \cite{BracherEtAl2021} suggest to evaluate these forecasts via a weighted interval score that corresponds to a weighted sum of the standard interval score at different coverage levels. This weighted interval score could similarly be decomposed using the methods proposed herein, allowing for an assessment of the conditional calibration of these collections of prediction intervals. As the number of prediction intervals increases, the decomposition terms tend towards those for the CRPS decomposition proposed by \cite{ArnoldEtAl2023}.

More generally, the desirable theoretical properties of the decomposition are a result of the universality of isotonic distributional regression, which means that IDR is optimal (in-sample) with respect to a large class of scoring rules \citep[Theorem 2]{HenziEtAl2021}. The quantile score, interval score, weighted interval score, and CRPS are all examples of scoring rules within this class. This class additionally contains weighted versions of the CRPS \citep{GneitingRanjan2011}. By applying an analogous decomposition to the threshold-weighted CRPS, we would obtain a framework with which to assess conditional notions of tail calibration of probabilistic forecasts \citep{AllenEtAl2025}.

Finally, while score decompositions assess auto-calibration and isotonic calibration, which are generally strong notions of forecast calibration, the decompositions can be made stronger by conditioning on larger information sets \citep{EhmOvcharov2017,AllenEtAl2023}. This allows for an assessment of additional variables that may induce conditional forecast biases. Future work could apply these decompositions to the interval score, and study to what extent isotonic (distributional) regression can be used to obtain meaningful estimates of these decomposition terms in practice.

\bibliographystyle{apalike}
\bibliography{biblio.bib} 

\begin{thebibliography}{}

\bibitem[Allen et~al., 2023]{AllenEtAl2023}
Allen, S., Ferro, C. A.~T., and Kwasniok, F. (2023).
\newblock A conditional decomposition of proper scores: Quantifying the sources
  of information in a forecast.
\newblock {\em Quarterly Journal of the Royal Meteorological Society},
  149:1704--1725.

\bibitem[Allen et~al., 2025]{AllenEtAl2025}
Allen, S., Koh, J., Segers, J., and Ziegel, J. (2025).
\newblock Tail calibration of probabilistic forecasts.
\newblock {\em Journal of the American Statistical Society (Theory \&
  Methods)}.
\newblock To appear.

\bibitem[Arnold et~al., 2024]{ArnoldEtAl2023}
Arnold, S., Walz, E.-M., Ziegel, J., and Gneiting, T. (2024).
\newblock Decompositions of the mean continuous ranked probability score.
\newblock {\em Electronic Journal of Statistics}, 18:4992--5044.

\bibitem[Arnold and Ziegel, 2025]{ArnoldZiegel2023}
Arnold, S. and Ziegel, J. (2025).
\newblock Isotonic conditional laws.
\newblock {\em Bernoulli}, 31:1140--1159.

\bibitem[Bashaykh, 2022]{Bashaykh2022}
Bashaykh, H. (2022).
\newblock {\em Statistical Assessment of Forecast Calibration}.
\newblock PhD thesis, University of Exeter.

\bibitem[Bentzien and Friederichs, 2014]{BentzienFriederichs2014}
Bentzien, S. and Friederichs, P. (2014).
\newblock Decomposition and graphical portrayal of the quantile score.
\newblock {\em Quarterly Journal of the Royal Meteorological Society},
  140:1924--1934.

\bibitem[Bracher et~al., 2021]{BracherEtAl2021}
Bracher, J., Ray, E.~L., Gneiting, T., and Reich, N.~G. (2021).
\newblock Evaluating epidemic forecasts in an interval format.
\newblock {\em PLoS Computational Biology}, 17:e1008618.

\bibitem[Br{\"o}cker, 2009]{Brocker2009}
Br{\"o}cker, J. (2009).
\newblock Reliability, sufficiency, and the decomposition of proper scores.
\newblock {\em Quarterly Journal of the Royal Meteorological Society},
  135:1512--1519.

\bibitem[Candille and Talagrand, 2005]{CandilleTalagrand2005}
Candille, G. and Talagrand, O. (2005).
\newblock Evaluation of probabilistic prediction systems for a scalar variable.
\newblock {\em Quarterly Journal of the Royal Meteorological Society},
  131:2131--2150.

\bibitem[Christoffersen, 1998]{Christoffersen1998}
Christoffersen, P.~F. (1998).
\newblock Evaluating interval forecasts.
\newblock {\em International Economic Review}, pages 841--862.

\bibitem[{COVID-19 Forecast Hub}, 2020]{Covid19Hub}
{COVID-19 Forecast Hub} (2020).
\newblock {COVID-19 Forecast Hub}.
\newblock UMass-Amherst Influenza Forecasting Center of Excellence. Accessible
  online at \url{https://github.com/reichlab/covid19-forecast-hub}. Last
  accessed 22 June 2025.

\bibitem[DeGroot and Fienberg, 1982]{DeGrootFeinberg1983}
DeGroot, M. and Fienberg, S. (1982).
\newblock Assessing probability assessors: Calibration and refinement.
\newblock {\em Statistical Decision Theory and Related Topics III}, pages
  291--314.

\bibitem[Dimitriadis et~al., 2021]{DimitriadisEtAl2021}
Dimitriadis, T., Gneiting, T., and Jordan, A.~I. (2021).
\newblock Stable reliability diagrams for probabilistic classifiers.
\newblock {\em Proceedings of the National Academy of Sciences},
  118:e2016191118.

\bibitem[Dimitriadis et~al., 2024]{Dimitriadis_2023}
Dimitriadis, T., Gneiting, T., Jordan, A.~I., and Vogel, P. (2024).
\newblock Evaluating probabilistic classifiers: The triptych.
\newblock {\em International Journal of Forecasting}, 40:1101--1122.

\bibitem[Ehm and Ovcharov, 2017]{EhmOvcharov2017}
Ehm, W. and Ovcharov, E. (2017).
\newblock Bias-corrected score decomposition for generalized quantiles.
\newblock {\em Biometrika}, 104:473--480.

\bibitem[Fanaee-T and Gama, 2014]{FanaeeTGama2014}
Fanaee-T, H. and Gama, J. (2014).
\newblock Event labeling combining ensemble detectors and background knowledge.
\newblock {\em Progress in Artificial Intelligence}, 2:113--127.

\bibitem[Feldman et~al., 2021]{FeldmanEtAl2021}
Feldman, S., Bates, S., and Romano, Y. (2021).
\newblock Improving conditional coverage via orthogonal quantile regression.
\newblock {\em Advances in neural information processing systems},
  34:2060--2071.

\bibitem[Fontana et~al., 2023]{FontanaEtAl2023}
Fontana, M., Zeni, G., and Vantini, S. (2023).
\newblock Conformal prediction: a unified review of theory and new challenges.
\newblock {\em Bernoulli}, 29:1--23.

\bibitem[Gneiting, 2011]{Gneiting2011}
Gneiting, T. (2011).
\newblock Making and evaluating point forecasts.
\newblock {\em Journal of the American Statistical Association}, 106:746--762.

\bibitem[Gneiting et~al., 2007]{GneitingEtAl2007}
Gneiting, T., Balabdaoui, F., and Raftery, A.~E. (2007).
\newblock Probabilistic forecasts, calibration and sharpness.
\newblock {\em Journal of the Royal Statistical Society Series B: Statistical
  Methodology}, 69:243--268.

\bibitem[Gneiting and Raftery, 2007]{GneitingRaftery2007}
Gneiting, T. and Raftery, A.~E. (2007).
\newblock Strictly proper scoring rules, prediction, and estimation.
\newblock {\em Journal of the American Statistical Association}, 102:359--378.

\bibitem[Gneiting and Ranjan, 2011]{GneitingRanjan2011}
Gneiting, T. and Ranjan, R. (2011).
\newblock Comparing density forecasts using threshold-and quantile-weighted
  scoring rules.
\newblock {\em Journal of Business \& Economic Statistics}, 29:411--422.

\bibitem[Gneiting and Ranjan, 2013]{GneitingRanjan2013}
Gneiting, T. and Ranjan, R. (2013).
\newblock Combining predictive distributions.
\newblock {\em Electronic Journal of Statistics}, 7:1747--1782.

\bibitem[Gneiting and Resin, 2023]{GneitingResin2023}
Gneiting, T. and Resin, J. (2023).
\newblock Regression diagnostics meets forecast evaluation: Conditional
  calibration, reliability diagrams, and coefficient of determination.
\newblock {\em Electronic Journal of Statistics}, 17:3226--3286.

\bibitem[Gneiting et~al., 2023]{GneitingEtAl2023}
Gneiting, T., Wolffram, D., Resin, J., Kraus, K., Bracher, J., Dimitriadis, T.,
  Hagenmeyer, V., Jordan, A.~I., Lerch, S., Phipps, K., and Schienle, M.
  (2023).
\newblock Model diagnostics and forecast evaluation for quantiles.
\newblock {\em Annual Review of Statistics and Its Application}, 10:597--621.

\bibitem[Henzi, 2021]{isodistrreg}
Henzi, A. (2021).
\newblock isodistreg: Isotonic distributional regression ({IDR}):
  {D}istributional regression under stochastic order restrictions for numeric
  and binary response variables and partially ordered covariates.
\newblock {\em The Comprehensive R Archive Network}.

\bibitem[Henzi et~al., 2021]{HenziEtAl2021}
Henzi, A., Ziegel, J.~F., and Gneiting, T. (2021).
\newblock Isotonic distributional regression.
\newblock {\em Journal of the Royal Statistical Society Series B: Statistical
  Methodology}, 83:963--993.

\bibitem[Jordan et~al., 2022]{JordanEtAl2022}
Jordan, A.~I., M{\"u}hlemann, A., and Ziegel, J.~F. (2022).
\newblock Characterizing the optimal solutions to the isotonic regression
  problem for identifiable functionals.
\newblock {\em Annals of the Institute of Statistical Mathematics},
  74:489--514.

\bibitem[Lei et~al., 2018]{LeiEtAl2018}
Lei, J., G’Sell, M., Rinaldo, A., Tibshirani, R.~J., and Wasserman, L.
  (2018).
\newblock Distribution-free predictive inference for regression.
\newblock {\em Journal of the American Statistical Association},
  113:1094--1111.

\bibitem[Makridakis et~al., 2020]{Makridakis2020}
Makridakis, S., Spiliotis, E., and Assimakopoulos, V. (2020).
\newblock The {M4} competition: 100,000 time series and 61 forecasting methods.
\newblock {\em International Journal of Forecasting}, 36:54--74.

\bibitem[Makridakis et~al., 2022]{Makridakis2022}
Makridakis, S., Spiliotis, E., Assimakopoulos, V., Chen, Z., Gaba, A., Tsetlin,
  I., and Winkler, R.~L. (2022).
\newblock The {M5} uncertainty competition: Results, findings and conclusions.
\newblock {\em International Journal of Forecasting}, 38:1365--1385.

\bibitem[Matheson and Winkler, 1976]{MathesonWinkler1976}
Matheson, J.~E. and Winkler, R.~L. (1976).
\newblock Scoring rules for continuous probability distributions.
\newblock {\em Management Science}, 22:1087--1096.

\bibitem[Mitchell, 2020]{Mitchell2019}
Mitchell, K. (2020).
\newblock {\em Score Decompositions in Forecast Verification}.
\newblock PhD thesis, University of Exeter.

\bibitem[Mosteller, 1995]{Mosteller1995}
Mosteller, F. (1995).
\newblock The {Tennessee} study of class size in the early school grades.
\newblock {\em The Future of Children}, 5:113--127.

\bibitem[Murphy, 1973]{Murphy1973}
Murphy, A.~H. (1973).
\newblock A new vector partition of the probability score.
\newblock {\em Journal of Applied Meteorology and Climatology}, 12:595--600.

\bibitem[Papadopoulos et~al., 2008]{PapadopoulosEtAl2008}
Papadopoulos, H., Gammerman, A., and Vovk, V. (2008).
\newblock Normalized nonconformity measures for regression conformal
  prediction.
\newblock In {\em Proceedings of the IASTED International Conference on
  Artificial Intelligence and Applications (AIA 2008)}, pages 64--69.

\bibitem[Romano et~al., 2019]{RomanoEtAl2019}
Romano, Y., Patterson, E., and Candès, E.~J. (2019).
\newblock Conformalized quantile regression.
\newblock {\em Advances in neural information processing systems},
  32:3543--3553.

\bibitem[Saerens, 2000]{Saerens2000}
Saerens, M. (2000).
\newblock Building cost functions minimizing to some summary statistics.
\newblock {\em IEEE Transactions on Neural Networks}, 11:1263--1271.

\bibitem[Sanders, 1963]{Sanders1963}
Sanders, F. (1963).
\newblock On subjective probability forecasting.
\newblock {\em Journal of Applied Meteorology and Climatology}, 2:191--201.

\bibitem[Sebasti{\'a}n et~al., 2024]{SebastianEtAl2024}
Sebasti{\'a}n, C., Gonz{\'a}lez-Guill{\'e}n, C.~E., and Juan, J. (2024).
\newblock Enhancing reliability in prediction intervals using point
  forecasters: Heteroscedastic quantile regression and width-adaptive conformal
  inference.
\newblock {\em arXiv preprint arXiv:2406.14904}.

\bibitem[Shafer and Vovk, 2008]{ShaferVovk2008}
Shafer, G. and Vovk, V. (2008).
\newblock A tutorial on conformal prediction.
\newblock {\em Journal of Machine Learning Research}, 9:371--421.

\bibitem[Singh, 2016]{facebook}
Singh (2016).
\newblock Facebook comment volume dataset.
\newblock
  \url{https://archive.ics.uci.edu/dataset/363/facebook+comment+volume+dataset}.

\bibitem[Thomson, 1979]{Thomson1979}
Thomson, W. (1979).
\newblock Eliciting production possibilities from a well-informed manager.
\newblock {\em Journal of Economic Theory}, 20:360--380.

\bibitem[Tsyplakov, 2011]{Tsyplakov2011}
Tsyplakov, A. (2011).
\newblock Evaluating density forecasts: A comment.
\newblock {\em Available at SSRN 1907799}.

\bibitem[Vovk et~al., 2022]{VovkEtAl2022}
Vovk, V., Gammerman, A., and Shafer, G. (2022).
\newblock {\em Algorithmic Learning in a Random World}.
\newblock Springer, 2nd edition.

\bibitem[Vovk et~al., 2018]{VovkEtAl2018}
Vovk, V., Nouretdinov, I., Manokhin, V., and Gammerman, A. (2018).
\newblock Cross-conformal predictive distributions.
\newblock In {\em Conformal and Probabilistic Prediction and Applications},
  pages 37--51. Proceedings of Machine Learning Research.

\bibitem[Waghmare and Ziegel, 2025]{WaghmareZiegel2025}
Waghmare, K. and Ziegel, J. (2025).
\newblock Proper scoring rules for estimation and forecast evaluation.
\newblock {\em Annual Review of Statistics and Its Application}.

\bibitem[Winkler, 1972]{Winkler1972}
Winkler, R.~L. (1972).
\newblock A decision-theoretic approach to interval estimation.
\newblock {\em Journal of the American Statistical Association}, 67:187--191.

\end{thebibliography}

\appendix
\section{Appendices}

\subsection{Proofs}\label{app:proofs}

\begin{proof}[Proof of Proposition \ref{prop:cal}]
For part (i), it suffices to take expectations in \eqref{eq:cal_con2}. Part (ii) follows from part (i) and (iii). Part (iii) follows from \citet[Proposition 6.3]{ArnoldZiegel2023} and the fact that $L$ and $U$ are $\Aa(L,U)$-measurable.
For the reverse directions, see \citet[Examples 5.1 and 5.2]{ArnoldZiegel2023}. The condition in part (iv) implies that $Q_\beta(Y \mid L,U) = Q_\beta(Y \mid \Aa(L,U))$; see the last paragraph of Section 6.1 in \citet{ArnoldZiegel2023}. In this case, the definitions of auto-calibration and isotonic calibration at \eqref{eq:cal_con} and \eqref{eq:cal_iso} are equivalent.
\end{proof}

\begin{proof}[Proof of Proposition \ref{prop:decomp_auto} and \ref{prop:decomp_IDR_1}]

   We only prove the properties of the decomposition assessing isotonic interval calibration (Proposition \ref{prop:decomp_IDR_1}) as the statements for the decomposition evaluating auto-calibration follow, since $\sigma(L,U)$ is also a $\sigma$-lattice and all $\mathcal{A}(L,U)$-measurable random variables are also $\sigma(L,U)$-measurable.
   
   We recall Proposition 6.3 in \cite{ArnoldEtAl2023} on the characterisation of conditional quantiles given a general $\sigma$-lattice. It states that
   for all $\alpha\in (0,1)$, a random variable belongs to the conditional quantile $Q_{\alpha}(Y\,|\,\mathcal{A})$ given some $\sigma$-lattice $\mathcal{A}$ if and only if it minimises the expected quantile score $\mathbb{E}[\mathrm{QS}_{\alpha}(X,Y)]$ over all $\mathcal{A}$-measurable random variables $X$.
   
   Since $\mathrm{IS}_{\alpha}([L,U],Y)=\frac{2}{\alpha}[\mathrm{QS}_{\alpha/2}(L,Y)+\mathrm{QS}_{1-\alpha/2}(U,Y)]$, the discrimination term can be written as
   \begin{align}
       \mathrm{DSC}_{\mathrm{I}} &=\frac{2}{\alpha}\,\Big[ \mathbb{E}[\mathrm{QS}_{\frac{\alpha}{2}}(q_{\frac{\alpha}{2}}(Y),Y)]-\mathbb{E}[\mathrm{QS}_{\frac{\alpha}{2}}(q_{\frac{\alpha}{2}}(Y\,|\,\mathcal{A}(L,U)),Y)]\Big]\\
       &+ \frac{2}{\alpha}\,\Big[\mathbb{E}[\mathrm{QS}_{1-\frac{\alpha}{2}}(q_{1-\frac{\alpha}{2}}(Y),Y)]-\mathbb{E}[\mathrm{QS}_{1-\frac{\alpha}{2}}(q_{1-\frac{\alpha}{2}}(Y\,|\,\mathcal{A}(L,U)),Y)]\Big] \nonumber.
   \end{align}
   It satisfies $\mathrm{DSC}_{\mathrm{I}}\geq 0$ because $q_{\frac{\alpha}{2}}(Y)$ and $q_{1-\frac{\alpha}{2}}(Y)$ are both $\mathcal{A}(L,U)$-measurable, since they are constant. Moreover, both terms in the large square brackets are non-negative. Therefore, $\mathrm{DSC}_{\mathrm{I}}=0$ if and only if $q_{\frac{\alpha}{2}}(Y\,|\,\mathcal{A}(L,U)) \in Q_{\frac{\alpha}{2}}(Y)$ and 
   $q_{1-\frac{\alpha}{2}}(Y\,|\,\mathcal{A}(L,U)) \in Q_{1-\frac{\alpha}{2}}(Y)$ almost surely.
   
   The miscalibration term can be rewritten in a similar way and since $L$ and $U$ are both $\mathcal{A}(L,U)$-measurable by definition, it follows that $\mathrm{MCB}_{\mathrm{I}}\geq 0$ with equality if and only if $L \in Q_{\frac{\alpha}{2}}(Y\,|\,\mathcal{A}(L,U))$ and $U \in Q_{1-\frac{\alpha}{2}}(Y\,|\,\mathcal{A}(L,U))$ almost surely.\\

\end{proof}

\begin{proof}[Proof of Proposition \ref{prop:IDR_auto}]

    Since $\mathcal{A}(L,U)\subseteq \sigma(L,U)$, a conditional quantile given $\mathcal{A}(L,U)$ is $\sigma(L,U)$-measurable. Therefore,  
    \begin{align}
      \mathbb{E}[\mathrm{QS}_{\frac{\alpha}{2}}(q_{\frac{\alpha}{2}}(Y\,|\,L,U),Y)] &\leq \mathbb{E}[\mathrm{QS}_{\frac{\alpha}{2}}(q_{\frac{\alpha}{2}}(Y\,|\,\mathcal{A}
      (L,U)),Y)]\\
      \mathbb{E}[\mathrm{QS}_{1-\frac{\alpha}{2}}(q_{1-\frac{\alpha}{2}}(Y\,|\,L,U),Y)] &\leq \mathbb{E}[\mathrm{QS}_{1-\frac{\alpha}{2}}(q_{1-\frac{\alpha}{2}}(Y\,|\,\mathcal{A}
      (L,U)),Y)]\nonumber
    \end{align}
    and therefore $\mathrm{MCB}_{\mathrm{A}}\geq \mathrm{MCB}_{\mathrm{I}}$.\\

     Since the interval score is non-negative by definition, we therefore have $\mathbb{E}[\mathrm{IS}_{\alpha}(L,U,Y)]\geq \mathrm{MCB}_{\mathrm{A}}\geq \mathrm{MCB}_{\mathrm{I}}.$
\end{proof}

\begin{proof}[Proof of Proposition \ref{prop:decomp_IDR}]

For the proof of part (i), suppose that $\widehat{\mathrm{DCB}}_I = 0$. This implies that $\bar{q}_{\alpha/2}$ is a minimiser of the expression at \eqref{eq:quantile_average}, hence $\bar{q}_{\alpha/2} \ge \tilde{q}_{\alpha/2,i}$ for all $i=1,\dots,n$. The lower isotonic quantiles $\tilde{q}_{\alpha/2,i}$ induce a partition of $\{1,\dots,n\}$ such that $\tilde{q}_{\alpha/2,i}$ is the lower quantile of the empirical distribution of $y_j$ for $j \in B$, where $B \subset \{1,\dots,n\}$ is the partition element that contains $i$ \citep[proof of Theorem 2.2]{HenziEtAl2021}. The partitions do not need to be unique. The stochastic order of the empirical distributions implies that $\bar{q}_{\alpha/2} \le \tilde{q}_{\alpha/2,i}$ for all $i=1,\dots,n$. This implies that $\bar{q}_{\alpha/2} = \tilde{q}_{\alpha/2,i}$ for all $i=1,\dots,n$.

Part (ii) is immediate.

\end{proof}

\subsection{Generalisations}\label{app:extensions}

\subsubsection{Extension to non-central prediction intervals}

The interval score is constructed to evaluate central prediction intervals. However, it can readily be extended to evaluate non-central prediction intervals by combining the quantile score at different quantile levels. More generally, a non-central interval score at levels $\alpha_1, \alpha_2 \in (0, 1)$, with $\alpha_1 < \alpha_2$, could be defined as
\begin{align*}
    \mathrm{IS}_{\alpha_1, \alpha_2}([\ell, u], y) &= \frac{1}{\alpha_1} \mathrm{QS}_{\alpha_1}(\ell, y) + \frac{1}{1 - \alpha_2} \mathrm{QS}_{\alpha_2}(u, y) \\
    &= | u - \ell | + \frac{1}{\alpha_1} \one\{y < \ell\} (\ell - y) + \frac{1}{1 - \alpha_2} \one\{y > u\} (y - u).
\end{align*}
The standard interval score is a particular example of this when $\alpha_1 = \alpha/2$ and $\alpha_2 = 1 - \alpha/2$. The notions of interval forecast calibration in Section \ref{sec:cali} can similarly be generalised to the non-central case.

\begin{defin}
    A non-central interval forecast $[L, U]$ at levels $\alpha_1,\alpha_2$ is \emph{unconditionally calibrated} if 
    \begin{equation*}
        \Q(Y < L) \le \alpha_1 \le \Q(Y \le L) \quad \text{and} \quad \Q(Y < U) \le \alpha_2 \le \Q(Y \le U).
    \end{equation*}
\end{defin}

This simplifies to $\Q(Y \le L) = \alpha_1$ and $\Q(Y \le U) = \alpha_2$ when $\Q(Y \in \{L, U\}) = 0$.

\begin{defin}
    A non-central interval forecast $[L, U]$ at levels $\alpha_1,\alpha_2$ is \emph{auto-calibrated} if
    \begin{equation*}
        L \in Q_{\alpha_1}(Y \mid L, U) \quad \text{and} \quad U \in Q_{\alpha_2}(Y \mid L, U).
    \end{equation*}
\end{defin}

\begin{defin}
    A non-central interval forecast $[L, U]$ at levels $\alpha_1,\alpha_2$ is \emph{isotonically calibrated} if 
    \begin{equation*}
        L \in Q_{\alpha_1}(Y \mid \Aa(L, U)) \quad \text{and} \quad U \in Q_{\alpha_2}(Y \mid \Aa(L, U)),
    \end{equation*}
    where $\Aa(L, U)$ denotes the $\sigma$-lattice generated by $L$ and $U$. 
\end{defin}

The relationships between the different notions of calibration listed in Proposition \ref{prop:cal} extend readily to these more general definitions of calibration for non-central prediction intervals. The auto- and isotonic calibration of the non-central interval forecasts can then be assessed analogously to as in Section \ref{sec:decomps}, by decomposing the non-central interval score by comparing the expected score to the expected score of the non-central intervals obtained from the unconditional and conditional distributions of $Y$ at the relevant levels. All of the desirable properties of the decomposition terms extend naturally to the non-central case.

\subsubsection{Extension to generalised interval scores}

As mentioned in Section \ref{sec:decomps}, the interval score is strictly consistent for central prediction intervals, which follows from the representation of the interval score as a weighted sum of two quantile scores. However, the quantile score is not the only scoring function that is strictly consistent for quantiles. In particular, the quantile score at \eqref{eq:qs} can be extended to the generalised quantile score (or generalised piecewise linear loss), which takes the form
\[
    \mathrm{GQS}_{\alpha,g}(x, y) = (\one\{y \le x\} - \alpha)(g(x) - g(y)),
\]
for some non-decreasing increasing function $g : \R \to \R$ \citep{Thomson1979,Saerens2000}.

A generalised interval score can therefore also be defined as 
\begin{align*}
    \mathrm{GIS}_{\alpha, g}([\ell, u], y) &= \frac{2}{\alpha} \left[ \mathrm{GQS}_{\frac{\alpha}{2},g}(\ell, y) + \mathrm{GQS}_{1 - \frac{\alpha}{2},g}(u, y) \right] \\
    &= | g(u) - g(\ell) | + \frac{2}{\alpha} \one\{y > u\} (g(y) - g(u)) + \frac{2}{\alpha} \one\{y < \ell\} (g(\ell) - g(y)),
\end{align*}
which is also consistent for the central $(1-\alpha)$-prediction level. The standard quantile score and interval score are recovered when $g$ is the identity function. The generalised interval score essentially transforms the observation and the prediction interval before calculating the standard interval score, with the transformation allowing the events $\{y > u\}$ and $\{y < \ell\}$ to be penalised differently.

The generalised interval score can be decomposed analogously to the standard interval score. In particular, we can replace the interval score in the decomposition terms with the generalised interval score, and apply this to the quantiles obtained from isotonic distributional regression. If the function $g$ is strictly increasing, then the generalised quantile score is strictly consistent for quantiles, and we recover exactly the same theoretical properties as outlined in Section \ref{sec:decomps}.

\subsection{Orderings}\label{app:orderings}

To calculate the sample level decomposition of the interval score in Section \ref{sec:iso_dcmp_sample}, we define a partial order on the space of interval forecasts. This is necessary for the implementation of IDR. We choose an ordering that defines an interval forecast as larger than another if both the lower and upper bounds are no smaller than those of the other forecast, with at least one bound being strictly larger. This is similar in essence to stochastic ordering of predictive distributions. While other partial orders could also be employed, not all orderings will yield interval score decomposition terms with the desirable properties listed in Section \ref{sec:iso_dcmp_sample}. 

IDR is designed so that a larger covariate value, according to the chosen partial order, results in a larger predictive distribution, according to the stochastic order on distributions (denoted $\le_{\text{st}}$ in the following). In our context, the IDR predictive distributions uniquely minimise the average interval score over all $n$-tuples $(F_1, \dots, F_n)$ of distribution functions that satisfy $F_i \le_{\text{st}} F_j$ if $[\ell_i, u_i] \le [\ell_j, u_j]$ for $i,j \in \{1, \dots, n\}$, where $[\ell_1, u_1], \dots, [\ell_n, u_n]$ represent the original interval forecasts on which IDR is applied. We additionally have that the IDR prediction intervals uniquely minimise the average interval score over all $n$-tuples $([\ell^*_1, u^*_1], \dots, [\ell^*_n, u^*_n])$ of interval forecasts that satisfy $\ell^*_i \le \ell^*_j$ and $u^*_i \le u^*_j$ if $[\ell_i, u_i] \le [\ell_j, u_j]$ for $i,j \in \{1, \dots, n\}$. These statements hold regardless of the ordering chosen on the space of interval forecasts.

The theoretical guarantees in Proposition \ref{prop:decomp_IDR} rely on the fact that original interval forecasts satisfy the same order relations as the prediction intervals obtained from the IDR recalibrated predictive distributions. That is, for the partial order employed in Section \ref{sec:iso_dcmp_sample}, the $n$-tuple of original forecasts $([\ell_1, u_1], \dots, [\ell_n, u_n])$ satisfies $\ell_i \le \ell_j$ and $u_i \le u_j$ if $[\ell_i, u_i] \le [\ell_j, u_j]$, since this is simply the definition of the partial order. They therefore belong to the class of $n$-tuples over which IDR is the optimal solution. The average interval score of the IDR recalibrated prediction intervals is thus no worse than that of the original forecasts, which ensures that the miscalibration term of the isotonic decomposition is non-negative.

As an example of where this does not hold, consider the order on the space of interval forecasts defined by $[\ell_i, u_i] \le [\ell_j, u_j]$ if $(\ell_i + u_i)/2 \le (\ell_j + u_j)/2$. That is, one interval forecast is larger than another if its midpoint is larger. This corresponds to a total order rather than a partial order, which could circumvent the practical issues related to non-comparable prediction intervals discussed in Remark \ref{rem:ass}. However, with this ordering, $[\ell_i, u_i] \le [\ell_j, u_j]$ does not necessarily imply that $\ell_i \le \ell_j$ and $u_i \le u_j$. Hence, while the IDR predictive intervals are still optimal with respect to the average interval score over all intervals that satisfy $\ell^*_i \le \ell^*_j$ and $u^*_i \le u^*_j$ if $[\ell_i, u_i] \le [\ell_j, u_j]$ for $i,j \in \{1, \dots, n\}$, this criterion is not satisfied by the original forecasts. There is therefore no guarantee that the average interval score of the recalibrated prediction intervals is lower than the average interval score of the original interval forecasts, meaning the miscalibration term of the decomposition could be negative. It is therefore necessary to choose an order that aligns with the stochastic order of predictive distributions that is used in IDR.

\end{document}